\documentclass[conference]{ieeeconf}
\IEEEoverridecommandlockouts
\usepackage{cite}
\usepackage{amsmath,amssymb,amsfonts,amsthm}
\usepackage{graphicx}
\usepackage{textcomp}
\usepackage{xcolor}
\usepackage{colortbl}
\usepackage{multirow}
\usepackage{hhline}
\usepackage{algorithm}

\usepackage[noend]{algpseudocode}
\usepackage{comment}
\usepackage{bm}
\usepackage[T1]{fontenc}
\let\labelindent\relax
\usepackage{enumitem}

\newtheorem{definition}{Definition}

\newtheorem{theorem}{Theorem}

\newtheorem{proposition}[theorem]{Proposition}
\newtheorem{lemma}[theorem]{Lemma}
\theoremstyle{definition}
\newtheorem{example}{Example}

\newcommand{\rev}[1]{\textcolor{black}{#1}}

\title{\LARGE \bf On the Fragility of Efficiency: Uncoupled Learning with a Deviant }

\author{Vade Shah and Jason R. Marden
\thanks{This work is supported by the AFOSR (grant \#FA9550-25-1-0245), NASA (grant \#103215), and the NSF GRFP (grant \#2139319).}
\thanks{V. Shah and J. R. Marden are with the Department of Electrical and Computer Engineering at the University of California, Santa Barbara, CA.}%
}

\begin{document}

\maketitle 

\begin{abstract}
An important question in learning in games is whether players can learn to achieve efficient outcomes, those that maximize the sum of their utilities, without communication or coordination. Completely uncoupled learning algorithms, which use only each player's own past actions and utilities, can guarantee this in broad classes of games, but they typically require every player to follow the same algorithm. In this work, we ask whether one can design welfare-maximizing learning algorithms that are robust to a deviant player. To study this question, we take a well-studied welfare-maximizing learning algorithm as a testbed and consider what happens when every player except one follows it. First, we consider a deviant player who either (i) always plays the same action or (ii) selects their action uniformly at random. We show that such a player can significantly alter the long-run behavior of this algorithm, leading to states that are far from welfare-maximizing, and that no completely uncoupled learning algorithm is robust to such a player in every game. Next, we consider a strategic deviant player who acts in their own interest and show that this player can adopt a different learning algorithm to shift the long-run outcome in their favor. Moreover, we show that under any completely uncoupled welfare-maximizing learning algorithm, there is a game in which some player benefits from switching to a different completely uncoupled learning algorithm. These results show that completely uncoupled welfare-maximizing learning is fragile, even to a single deviant player, suggesting that robustness may require communication or coordination.
\end{abstract}
 
\section{Introduction}

Many engineering systems are made up of interconnected subsystems, such as turbines in a wind farm \cite{marden2013model}, vehicles in a traffic network \cite{roughgarden2002bad}, and sensors in an array \cite{zhu2013distributed}. In these systems, the collective choices of the subsystems determine their individual performance. In a wind farm, for example, each turbine slows the wind that reaches the turbines behind it, so one turbine's power depends on how much energy every other one chooses to extract \cite{porteagel2020wind}.

In systems like these, an engineer often seeks to design the behavior of the individual subsystems so as to optimize some system-wide metric. In the wind farm example, this might entail choosing the power extraction at each wind turbine so as to maximize total power extraction across the entire farm. However, engineers lack accurate models of the functional relationship between the choices of each individual turbine and the total power \cite{boersma2017tutorial}, so directly solving an optimization problem is not an option. Thus, one is faced with the task of optimizing an \emph{unknown} system-wide metric through \emph{subsystem-level} decisions.

Game-theoretic learning offers a natural framework for these kinds of problems \cite{fudenberg1998theory, marden2015game, young2009learning}. Here, the subsystems can be viewed as players in a game who repeatedly select actions and receive utilities that depend on the resulting action profile. An engineer then seeks to give each player a \emph{learning algorithm}, a rule that tells them what to play based on what they have observed so far, that leads the entire system to a desirable configuration. When each player observes only local information (i.e., their own actions and utilities), which is a realistic informational constraint in settings like wind farm control, designing such \emph{completely uncoupled} learning algorithms becomes challenging \cite{hart2003uncoupled, foster2006regret, babichenko2012completely}. Nevertheless, researchers have successfully designed different completely uncoupled learning algorithms that lead players to desirable configurations such as Nash equilibria \cite{marden2009payoff, young2009learning}, action profiles that are efficient (i.e., maximize \emph{welfare}, the sum of the players' utilities) \cite{marden2014achieving, menon2014collaborative}, and welfare-maximizing Nash equilibria \cite{pradelski2012learning}.

Embedded in nearly all existing results is a requirement that every player follow the exact same learning algorithm. In large or distributed systems like wind farms, this requirement can become problematic: individual subsystems may be faulty, act in their own interests, or even behave adversarially. Hence, a growing line of work has begun to study what can be achieved with completely uncoupled learning algorithms when some players do not learn \cite{jaleel2019robustness, akbar2023robustness} or learn differently from others \cite{akbar2023robustness, rajab2026log, shah2024learning, shah2025trial}. However, most of these results focus on specific learning algorithms or application domains, so a general understanding of how nonconformist learners, or \emph{deviants}, influence welfare-maximizing learning algorithms remains underdeveloped.

This work studies completely uncoupled learning algorithms designed to maximize welfare and asks how the presence of a deviant affects their guarantees. As a starting point, we consider the well-studied welfare-maximizing (WM) learning algorithm of \cite{marden2014achieving}, which provides a natural testbed for these questions. The WM learning algorithm is known to lead players to welfare-maximizing action profiles in a broad class of games known as interdependent games. We examine what happens when every player except one follows the WM learning algorithm, while the remaining deviant does something else.

We first consider a deviant who is either stubborn, meaning they always play the same action, or confused, meaning they select their action uniformly at random at every timestep; both behaviors are adapted from \cite{akbar2023robustness} and can be thought of as modeling a faulty subsystem. We show that with a confused deviant, the remaining players never settle on an action profile, and with a stubborn deviant, the action profile the players reach may even be far from maximizing the welfare of the remaining non-deviant players. We then show that this fragility is not specific to this algorithm: no completely uncoupled learning algorithm is robust to a single stubborn or confused deviant in every interdependent game.

Next, similar to \cite{akbar2023robustness, deng2019strategizing}, we consider a deviant who acts in their own interest. We show that by changing a single parameter of the WM learning algorithm, the deviant can shift the long-run outcome toward action profiles that benefit them, and that by switching to a different learning algorithm, they can guarantee an action profile that maximizes their own utility. Finally, we show that this fragility is again not specific to the WM learning algorithm: for any completely uncoupled learning algorithm that leads the players to welfare-maximizing action profiles, there is an interdependent game in which some player can increase their utility by switching to a different learning algorithm.

These results add to the growing literature on learning in games with heterogeneous players, and we envision that they will help guide the design of learning algorithms that remain robust to heterogeneous learners.

\section{Model}

In this section, we introduce our model of game-theoretic learning along with the welfare-maximizing learning algorithm presented in \cite{marden2014achieving}.

\subsection{Games}

Consider an $N$-player strategic form game $G$. The players are indexed by $i \in \{1, \dots, N\} \triangleq \mathcal{N}$. Each player $i$ selects an action $a_i$ from the finite set $\mathcal{A}_i$. The \emph{action profile} $a = (a_1, \dots, a_N)$ belongs to the set $\mathcal{A} \triangleq \prod_{i \in \mathcal{N}} \mathcal{A}_i$. Each player $i$ has a utility function $U_i : \mathcal{A} \to \mathbb{R}$; it is assumed (without loss of generality) that the utilities are normalized so that $\mathcal{U}_i \subset [0, 1)$ for all $i \in \mathcal{N}$, where $\mathcal{U}_i$ denotes the image of $U_i$. The \emph{welfare} of an action profile $a$ is defined as
\begin{equation}
    W(a) = \sum_{i \in \mathcal{N}} U_i(a);
\end{equation}
an action profile that maximizes $W$ is said to be \emph{welfare-maximizing}. We write $a = (a_{\mathcal{I}}, a_{-\mathcal{I}})$ to separate the actions of players in a proper subset $\mathcal{I} \subset \mathcal{N}$ from those not in $\mathcal{I}$. Accordingly, for a player $j \in \mathcal{N}$ and a fixed action $a_j \in \mathcal{A}_j$, we define the \emph{residual welfare} of the remaining players as
\begin{equation}
    W_{-j}(a_j, a_{-j}) = \sum_{i \in \mathcal{N} \setminus \{j\}} U_i(a_j, a_{-j}).
\end{equation}
An action profile $(a_j, a_{-j})$ is said to be \emph{residual-welfare-maximizing} if $a_{-j}$ maximizes $W_{-j}(a_j, \cdot)$.

\subsection{Learning Algorithms}

We consider the setting where a game $G$ is played repeatedly over timesteps $t \in \{1, 2, \dots\}$. In each timestep $t$, each player $i$ simultaneously chooses their action $a_i(t)$; then, each player $i$ receives their utility $U_i(a(t))$. Player $i$ chooses their action according to a probability distribution over their actions, $p_i(t) \in \Delta(\mathcal{A}_i)$, which is referred to as their \emph{strategy}; we write $p_i^{a_i}(t)$ for the probability that player $i$ selects action $a_i \in \mathcal{A}_i$. When the dependence is clear, we may omit $(t)$.

The focus of this paper is on \emph{learning algorithms}, mechanisms that generate strategies $p_i(t)$ for players to use to select their actions. Learning algorithms are characterized by the information they use to generate a strategy; in this work, they are restricted to a player's action and utility history. Accordingly, a learning algorithm $F_i$ is of the form
\begin{equation}
    p_i(t) = F_i \left( \left\{ a_i(\tau), U_i(a(\tau)) \right\}_{\tau = 1}^{t - 1} \right).
\end{equation}
Learning algorithms of this form are said to be \emph{completely uncoupled}, a well-known, highly informationally restrictive class of algorithms. Despite this restrictiveness, completely uncoupled algorithms have been shown to guarantee convergence to a variety of desirable configurations, including welfare-maximizing action profiles \cite{marden2014achieving, menon2014collaborative}.

\subsection{Welfare-Maximizing Learning}

This paper concerns learning algorithms that maximize welfare. As a testbed for studying the questions we seek to address, we present the welfare-maximizing (WM) learning algorithm proposed in \cite{marden2014achieving}. It works by equipping each player with a state variable $z_i = (m_i, \bar{a}_i, \bar{u}_i)$, where
\begin{itemize}
    \item $m_i$ is their \emph{mood}, which can be either content ($C$) or discontent ($D$).
    \item $\bar{a}_i \in \mathcal{A}_i$ is their \emph{baseline action}.
    \item $\bar{u}_i \in \mathcal{U}_i$ is their \emph{baseline utility}. 
\end{itemize}
Fix constants $b \in (0, 1]$, $c \geq N$, and $\varepsilon \in (0, 1)$. Under the WM learning algorithm, the players select actions and update their states according to the following:

\vspace{0.2cm}
\noindent \textbf{Action selection:} At time $t$, each player $i \in \mathcal{N}$ simultaneously selects their action as follows:

\begin{itemize}
    \item If $m_i = C$: select action $a_i(t)$ according to
    \begin{equation}
        p_i^{a_i}(t) = \begin{cases}
            \frac{\varepsilon^c}{\left\lvert \mathcal{A}_i \right\rvert - 1} & \text{for } a_i \neq \bar{a}_i \\
            1 - \varepsilon^c & \text{for } a_i = \bar{a}_i. 
        \end{cases}
        \label{eq:wm-action-C}
    \end{equation}
    \item If $m_i = D$: select action $a_i(t) \in \mathcal{A}_i$ uniformly at random.
\end{itemize}

\vspace{0.2cm}
\noindent \textbf{State update:} Each player $i \in \mathcal{N}$ receives their utility $u_i = U_i(a(t))$ and uses it to update their state $z_i$ as follows:
\begin{itemize}
    \item If $m_i = C$:
    \begin{itemize}
        \item If $(a_i, u_i) = (\bar{a}_i, \bar{u}_i)$, then
        \begin{equation}
            z_i \to (C, a_i, u_i).
            \label{eq:wm-update-C-baseline}
        \end{equation}
        \item Otherwise,
        \begin{equation}
            z_i \to \begin{cases}
                (C, a_i, u_i) & \text{with probability (w.p.) } \varepsilon^{b (1 - u_i)} \\
                (D, a_i, u_i)  & \text{w.p. } 1 - \varepsilon^{b (1 - u_i)}.
            \end{cases}
            \label{eq:wm-update-C-explore}
        \end{equation}
    \end{itemize}
    \item If $m_i = D$:
    \begin{equation}
        z_i \to \begin{cases}
            (C, a_i, u_i) & \text{ w.p. } \varepsilon^{b (1 - u_i)} \\
            (D, a_i, u_i)  & \text{ w.p. } 1 - \varepsilon^{b (1 - u_i)}.
        \end{cases}
        \label{eq:wm-update-D}
    \end{equation}
\end{itemize}
The action selection and state update procedures described above define the WM learning algorithm, which we write as $\texttt{WM}(b, c)$. The algorithm in \cite{marden2014achieving} corresponds to $b = 1$; we include the parameter $b$ to study players with heterogeneous parameters in Section~\ref{sec:different-params}. When all players follow $\texttt{WM}(b, c)$, this induces a Markov process over the state space
\begin{equation}\label{eq:original-state-space}
    \mathcal{Z} = \prod_{i \in \mathcal{N}} \; (\{C, D\} \times \mathcal{A}_i \times \mathcal{U}_i).
\end{equation}
A state $z$ is said to be \emph{stochastically stable} \cite{young1993evolution} if $\lim_{\varepsilon \to 0} \mu^{\varepsilon}(z) > 0$, where $\mu^{\varepsilon}$ denotes the stationary distribution of this process; this is the primary notion of convergence considered in this work. For brevity, we refer to the baseline action profile $\bar{a}$ of a stochastically stable state $z = (m, \bar{a}, \bar{u})$ as a \emph{stochastically stable action profile}.

The main result in \cite{marden2014achieving} applies to a broad class of games known as interdependent games:
\begin{definition}
    A game $G$ is said to be \emph{interdependent} if for any action profile $a \in \mathcal{A}$ and any nonempty proper subset of players $\emptyset \subset \mathcal{I} \subset \mathcal{N}$, there exists a player $k \notin \mathcal{I}$ and a choice of actions $a_{\mathcal{I}}' \in \prod_{i \in \mathcal{I}} \mathcal{A}_i$ such that $U_k(a_{\mathcal{I}}, a_{-\mathcal{I}})\neq U_k(a_{\mathcal{I}}', a_{-\mathcal{I}}).$
\end{definition}
\noindent Roughly speaking, this condition states that any subset of players can affect the utility of at least one player outside of that subset. Equipped with this definition, the main result regarding the WM algorithm is restated in Theorem \ref{thm:welfare-maximizing}:

\begin{theorem}\label{thm:welfare-maximizing}
    Let $G$ be a finite $N$-player interdependent game, and suppose that each player $i \in \mathcal{N}$ follows $\texttt{WM}(b, c)$ with $b \in (0, 1]$ and $c \geq N$. Then, every stochastically stable action profile is welfare-maximizing.
\end{theorem}

\noindent This result states that in the long run, a welfare-maximizing action profile is played with high probability as $\varepsilon \to 0$.

\section{Results}\label{sec:results}

The focus of this paper is on understanding how \emph{heterogeneity} across players affects their ability to learn to play welfare-maximizing action profiles. To this end, we consider the setting where one player $j \in \mathcal{N}$, a \emph{deviant}, uses a completely uncoupled learning algorithm $F_j$, while all other players $i \neq j$ use a common learning algorithm $F \neq F_j$. We start with the WM algorithm as a baseline choice for $F$, then generalize to completely uncoupled welfare-maximizing algorithms more broadly. The forthcoming results analyze the potential effects of a deviant's learning algorithm $F_j$.

\subsection{Confused or Stubborn Deviant}\label{sec:confused-stubborn}

Drawing from \cite{akbar2023robustness}, first consider two extremely simple behaviors for player $j$:

\begin{itemize}
    \item \textbf{Confused:} at each timestep $t$, player $j$ selects an action uniformly at random from $\mathcal{A}_j$.
    \item \textbf{Stubborn:} at each timestep $t$, player $j$ plays a fixed action $a_j^{\rm stb} \in \mathcal{A}_j$.
\end{itemize}

\noindent These behaviors can be thought of as deliberate strategies or as some notion of defectiveness in the player's behavior. If player $j$ is confused or stubborn while the remaining players employ the WM learning algorithm, then this induces a different Markov process over the state space
\begin{equation}
    \mathcal{Z}_{-j} = \prod_{i \in \mathcal{N} \setminus \{j\}} \; (\{C, D\} \times \mathcal{A}_i \times \mathcal{U}_i).
\end{equation}
When player $j$ is stubborn, the action profile associated with a state $z_{-j} = (m_{-j}, \bar{a}_{-j}, \bar{u}_{-j}) \in \mathcal{Z}_{-j}$ is $(a_j^{\rm stb}, \bar{a}_{-j})$, and stochastically stable action profiles are those associated with stochastically stable states. Our first result establishes characteristics of stochastically stable states in either setting.

\begin{proposition}\label{prop:confused-stubborn}
    Let $G$ be a finite $N$-player interdependent game, and suppose that each player $i \in \mathcal{N} \setminus \{ j \}$ follows $\texttt{WM}(b, c)$ with $b \in (0, 1]$ and $c \geq N$.
    \begin{enumerate}[label=(\alph*)]
        \item Suppose that player $j$ is confused. Then, in every stochastically stable state, every player $i \in \mathcal{N} \setminus \{ j \}$ is discontent.
        \item Suppose that player $j$ is stubborn with an action $a_j^{\rm stb} \in \mathcal{A}_j$ that appears in some welfare-maximizing action profile. Then, there may exist no stochastically stable action profile that is welfare-maximizing.
        \item Suppose that player $j$ is stubborn with an action $a_j^{\rm stb} \in \mathcal{A}_j$. Then, there may exist no stochastically stable action profile that is residual-welfare-maximizing.
    \end{enumerate}
\end{proposition}

Proposition \ref{prop:confused-stubborn}, proven in Appendix \ref{app:proof-confused-stubborn}, demonstrates that the presence of a confused or stubborn deviant can significantly alter the long-run behavior of the system. A confused deviant renders all-discontent states stochastically stable, implying that every action profile is visited equally often in the long run. A stubborn deviant can not only prevent every welfare-maximizing action profile from being stochastically stable, but can also do the same for every residual-welfare-maximizing action profile. The following examples illustrate the outcomes described in parts (b) and (c) of Proposition \ref{prop:confused-stubborn}.

\begin{example}[Stubborn Deviant: Welfare]\label{ex:stubborn-welfare}
    Consider a two-player Prisoner's Dilemma with payoff matrix
    \begin{center}
        \renewcommand{\arraystretch}{1.4}\setlength{\tabcolsep}{4pt}
        \begin{tabular}{c|>{\centering\arraybackslash}p{1.0cm}|>{\centering\arraybackslash}p{1.0cm}|}
            \multicolumn{1}{c}{} & \multicolumn{1}{c}{$A$} & \multicolumn{1}{c}{$B$} \\
            \hhline{~|-|-|}
            $A$ & \cellcolor{blue!20}\rev{$\frac{5}{10}, \frac{5}{10}$} & \cellcolor{red!20}\rev{$0, \frac{8}{10}$} \\
            \hhline{~|-|-|}
            $B$ & \rev{$\frac{8}{10}, 0$} & \rev{$\frac{2}{10}, \frac{2}{10}$} \\
            \hhline{~|-|-|}
        \end{tabular}
    \end{center}
    Here, player $1$ selects the row, player $2$ selects the column, and the first entry in each cell is player $1$'s utility. Note that the game is interdependent and that $(A, A)$ (shaded blue) is the unique welfare-maximizing action profile.

    Suppose that player $1$ is stubborn with $a_1^{\rm stb} = A$, which is player $1$'s action in $(A, A)$, and that player $2$ follows $\texttt{WM}(b, c)$ with $b \in (0, 1]$ and $c \geq 2$. Since player $1$ never deviates from $A$, the dynamics reduce to those of player $2$ alone running the WM learning algorithm over the first row of the game, i.e., choosing between $U_2(A, A) = 0.5$ and $U_2(A, B) = 0.8$. Consequently, the only stochastically stable action profile is $(A, B)$ (shaded red)\rev{; see Appendix~\ref{app:proof-confused-stubborn}}. Thus, even though player $1$ is stubborn with a welfare-maximizing action, no stochastically stable action profile is welfare-maximizing.
\end{example}

\begin{example}[Stubborn Deviant: Residual Welfare]\label{ex:stubborn-residual}
    Consider the three-player game with $\mathcal{A}_1 = \mathcal{A}_2 = \mathcal{A}_3 = \{A, B\}$:
    \begin{center}
        \renewcommand{\arraystretch}{1.4}\setlength{\tabcolsep}{4pt}
        \begin{tabular}{c|>{\centering\arraybackslash}p{1.3cm}|>{\centering\arraybackslash}p{1.3cm}|}
            \multicolumn{3}{c}{$a_3 = A$} \\
            \multicolumn{1}{c}{} & \multicolumn{1}{c}{$A$} & \multicolumn{1}{c}{$B$} \\
            \hhline{~|-|-|}
            $A$ & \rev{$\frac{1}{10}, 0, \frac{1}{10}$} & \rev{$0, \frac{1}{10}, 0$} \\
            \hhline{~|-|-|}
            $B$ & \rev{$0, \frac{1}{10}, 0$} & \rev{$\frac{1}{10}, 0, \frac{1}{10}$} \\
            \hhline{~|-|-|}
        \end{tabular}
        \hspace{0.4cm}
        \begin{tabular}{c|>{\centering\arraybackslash}p{1.3cm}|>{\centering\arraybackslash}p{1.3cm}|}
            \multicolumn{3}{c}{$a_3 = B$} \\
            \multicolumn{1}{c}{} & \multicolumn{1}{c}{$A$} & \multicolumn{1}{c}{$B$} \\
            \hhline{~|-|-|}
            $A$ & \rev{$0, \frac{1}{10}, 0$} & \cellcolor{blue!20}\rev{$0, \frac{9}{10}, \frac{1}{10}$} \\
            \hhline{~|-|-|}
            $B$ & \rev{$\frac{1}{10}, 0, \frac{1}{10}$} & \cellcolor{red!20}\rev{$\frac{1}{10}, \frac{1}{10}, 0$} \\
            \hhline{~|-|-|}
        \end{tabular}
    \end{center}
    Here, player $1$ selects the row, player $2$ selects the column, player $3$ selects the matrix, and the $i$-th entry in each cell is the utility of player $i$. Notice that this game is interdependent, so when every player follows $\texttt{WM}(b, c)$ with $b \in (0, 1]$ and $c \geq 3$, the only stochastically stable action profile is $(A, B, B)$ (shaded blue). Moreover, if player 3 is stubborn with $a_3^{\rm stb} = B$ while players 1 and 2 continue to follow $\texttt{WM}(b, c)$, then $(A, B, B)$ is also residual-welfare-maximizing. However, in this case, player 1's utility does not depend on player 2's action; thus, the WM learning algorithm will lead player 1 to choose the action that maximizes their own utility (in this case, $a_1 = B$). It can thus be shown \rev{(see Appendix~\ref{app:proof-confused-stubborn})} that the only stochastically stable action profile is $(B, B, B)$ (shaded red), which is not residual-welfare-maximizing.
\end{example}

These results demonstrate that the WM learning algorithm is not robust to a stubborn or confused deviant. This raises the question: can one construct a different completely uncoupled welfare-maximizing learning algorithm that is robust to confused or stubborn deviants? It turns out that these examples are sufficient to conclude that this is impossible:

\begin{theorem}\label{thm:no-robust}
    There exists no completely uncoupled learning algorithm $F$ such that, in every finite interdependent game, whenever every player except one player $j \in \mathcal{N}$ follows $F$, the players are led to\footnote{Akin to \cite{marden2014achieving}, we use the phrase ``lead to'' to mean either convergence, almost sure convergence, or convergence in the sense of stochastic stability; a player's \emph{long-run utility} is their utility at the action profile they are led to.}
    \begin{enumerate}[label=(\alph*)]
        \item only welfare-maximizing action profiles if player $j$ is confused,
        \item only welfare-maximizing action profiles if player $j$ is stubborn with an action that appears in some welfare-maximizing action profile, or
        \item only residual-welfare-maximizing action profiles if player $j$ is stubborn.
    \end{enumerate}
\end{theorem}

\begin{proof}
    Statement (a) follows trivially: in the Prisoner's Dilemma from Example \ref{ex:stubborn-welfare}, suppose that player 2 is confused. Regardless of player 1's learning algorithm, player 2 plays $A$ with probability $1/2$ at every timestep, so the players cannot be led only to the welfare-maximizing profile $(A, A)$.\footnote{\rev{Using an argument similar to that for statements (b) and (c), one can also construct an interdependent game in which no completely uncoupled learning algorithm leads the remaining players to an action profile that maximizes the expected residual welfare when player $j$ is confused.}}

    To prove statements (b) and (c), we proceed by contradiction. Suppose there exists a completely uncoupled learning algorithm $F$ satisfying (b) (resp. (c)), and consider the game $G$ from Example \ref{ex:stubborn-residual} with player 3 stubborn with $a_3^{\rm stb} = B$. Since $(A, B, B)$ is the unique welfare-maximizing (resp. residual-welfare-maximizing) action profile, $F$ must lead player 1 to $a_1 = A$. Now consider the game $G'$ obtained from $G$ by replacing the utilities at $(A, B, B)$ with $(0, 0, 0.1)$. One can verify that $G'$ is interdependent, that $(B, B, B)$ is welfare-maximizing in $G'$, and that every welfare-maximizing (resp. residual-welfare-maximizing) action profile in $G'$ with $a_3 = B$ has $a_1 = B$; thus, $F$ must lead player 1 to $a_1 = B$ in $G'$. However, once $a_3^{\rm stb} = B$ is fixed, player 1's utility in both $G$ and $G'$ is $0$ at $A$ and $0.1$ at $B$, regardless of player 2's action. Since $F$ is completely uncoupled, player 1's sequence of actions and utilities therefore has the same distribution in both games, so $F$ must lead player 1 to the same action in both games, which is a contradiction.
\end{proof}

\rev{The mechanism underlying these examples is that a stubborn deviant can break interdependence: in Example~\ref{ex:stubborn-residual}, once player 3 fixes $a_3 = B$, player 1's utility no longer depends on player 2's action, so player 1 simply learns to maximize their own utility.} Hence, for any completely uncoupled welfare-maximizing learning algorithm, a single deviant can significantly alter the long-run behavior of the system, even if their behavior is extremely simple.

\subsection{Strategic Deviant: Different Parameters}\label{sec:different-params}

While a confused or stubborn deviant can destabilize welfare-maximizing states, they cannot in general steer the system toward a desired configuration. This subsection asks whether a deviant can guide the other players toward action profiles that benefit the deviant. Perhaps the simplest such deviation is for player $j$ to follow the WM learning algorithm, but with different parameters:

\begin{proposition}\label{prop:heterogeneous-params}
    Let $G$ be a finite $N$-player interdependent game, and suppose that each player $i \in \mathcal{N} \setminus \{ j \}$ follows $\texttt{WM}(b, c)$ while player $j$ follows $\texttt{WM}(b_j, c_j)$, with $b, b_j \in (0, 1]$ and $c, c_j \geq N$. Then, every stochastically stable action profile maximizes
    $
        \sum_{i \in \mathcal{N} \setminus \{j\}} U_i + \frac{b_j}{b} U_j.
    $
\end{proposition}

\noindent Proposition \ref{prop:heterogeneous-params}, proven in Appendix \ref{app:proof-heterogeneous-params}, establishes that $c_j$ does not affect which action profiles can be stochastically stable, while $b_j$ reweights player $j$'s utility in the welfare function. Thus, player $j$ can increase $b_j$ to bias the stochastically stable states toward action profiles that yield them a higher utility, as shown in the following example:

\begin{example}[Biasing the Weighted Welfare]\label{ex:heterogeneous-b}
    Consider again the Prisoner's Dilemma from Example \ref{ex:stubborn-welfare}:
    \begin{center}
        \renewcommand{\arraystretch}{1.4}\setlength{\tabcolsep}{4pt}
        \begin{tabular}{c|>{\centering\arraybackslash}p{1.0cm}|>{\centering\arraybackslash}p{1.0cm}|}
            \multicolumn{1}{c}{} & \multicolumn{1}{c}{$A$} & \multicolumn{1}{c}{$B$} \\
            \hhline{~|-|-|}
            $A$ & \cellcolor{blue!20}\rev{$\frac{5}{10}, \frac{5}{10}$} & \rev{$0, \frac{8}{10}$} \\
            \hhline{~|-|-|}
            $B$ & \cellcolor{red!20}\rev{$\frac{8}{10}, 0$} & \rev{$\frac{2}{10}, \frac{2}{10}$} \\
            \hhline{~|-|-|}
        \end{tabular}
    \end{center}
    Suppose that player $2$ follows $\texttt{WM}(1/2, 2)$. If player $1$ also follows $\texttt{WM}(1/2, 2)$, then Theorem \ref{thm:welfare-maximizing} implies that the only stochastically stable action profile is $(A, A)$ (shaded blue). Now suppose that player $1$ instead follows $\texttt{WM}(1, 2)$, so that $b_1 / b = 2$. Then,
    $U_2(A, A) + 2 U_1(A, A) = 1.5$
    and
    $U_2(B, A) + 2 U_1(B, A) = 1.6$,
    so the only stochastically stable action profile becomes $(B, A)$ (shaded red), in which player $1$ receives their maximum utility.
\end{example}

\subsection{Strategic Deviant: Different Learning Algorithm}

Since $b_j \leq 1$, the extent to which player $j$ can bias the stochastically stable states in Proposition \ref{prop:heterogeneous-params} is limited. To overcome this, consider the following learning algorithm, inspired by \cite{pradelski2012learning} and referred to as the utility-maximizing (UM) learning algorithm. As in the WM learning algorithm, player $j$ maintains a state $z_j = (m_j, \bar{a}_j, \bar{u}_j)$. They use it to select their action and update their state as follows:

\vspace{0.2cm}
\noindent \textbf{Action selection:} At time $t$, player $j$ selects their action as follows:
\begin{itemize}
    \item If $m_j = C$: select action $a_j(t) = \bar{a}_j$.
    \item If $m_j = D$: select action $a_j(t) \in \mathcal{A}_j$ uniformly at random.
\end{itemize}

\vspace{0.2cm}
\noindent \textbf{State update:} After receiving their utility $u_j$, player $j$ updates their state as follows:
\begin{itemize}
    \item If $m_j = C$:
    \begin{itemize}
        \item If $u_j = \bar{u}_j$, then
        \begin{equation}
            z_j \to \begin{cases}
                (C, a_j, u_j) & \text{ w.p. } 1 - \varepsilon^{\Phi(u_j)} \\
                (D, a_j, u_j)  & \text{ w.p. } \varepsilon^{\Phi(u_j)}.
            \end{cases}
        \end{equation}
        \item Otherwise,
        \begin{equation}
            z_j \to \begin{cases}
                (C, a_j, u_j) & \text{ w.p. } \varepsilon^{\Psi(u_j)} \\
                (D, a_j, u_j)  & \text{ w.p. } 1 - \varepsilon^{\Psi(u_j)}.
            \end{cases}
        \end{equation}
    \end{itemize}
    \item If $m_j = D$, then
    \begin{equation}
        z_j \to \begin{cases}
            (C, a_j, u_j) & \text{ w.p. } \varepsilon^{\Psi(u_j)} \\
            (D, a_j, u_j)  & \text{ w.p. } 1 - \varepsilon^{\Psi(u_j)}.
        \end{cases}
    \end{equation}
\end{itemize}

Here, $\Phi: \mathbb{R} \to (0, N)$ and $\Psi: \mathbb{R} \to (0, 1)$ are assumed to be increasing and decreasing functions, respectively. Furthermore, they must satisfy
\begin{align}
    \Phi(u_j^*) - \Psi(u_j^*) &> N - 1 \; \text{ and } \label{eq:GH-cond-1} \\
    \Phi(u_j^* - \Delta_j) - \Psi(u_j^* - \Delta_j) &< 0, \label{eq:GH-cond-2}
\end{align}
where
\begin{equation*}
    u_j^* = \max_{a \in \mathcal{A}} \, U_j(a) \; \text{ and } \; 
    \Delta_j = \min_{\substack{a, a' \in \mathcal{A} \\ U_j(a) \neq U_j(a')}} \left\lvert U_j(a) - U_j(a') \right\rvert.
\end{equation*}
The action selection and state update procedures described above define the UM learning algorithm, which we write as $\texttt{UM}(\Phi, \Psi)$. If player $j$ employs the UM learning algorithm while the remaining players employ the WM learning algorithm, then this defines a new Markov process over $\mathcal{Z}$ \eqref{eq:original-state-space}. \rev{The interesting piece here is not the learning algorithm's structure itself, but rather the guarantee it provides when the remaining players follow the WM learning algorithm:}
\begin{proposition}\label{prop:strategic}
    Let $G$ be a finite $N$-player interdependent game, and suppose that each player $i \in \mathcal{N} \setminus \{ j \}$ follows $\texttt{WM}(b, c)$ with $b \in (0, 1]$ and $c \geq N$ while player $j$ follows $\texttt{UM}(\Phi, \Psi)$ with $\Phi$ and $\Psi$ satisfying \eqref{eq:GH-cond-1} and \eqref{eq:GH-cond-2}. Then, every stochastically stable action profile maximizes $U_j$.
\end{proposition}

\noindent Proposition \ref{prop:strategic}, proven in Appendix \ref{app:proof-strategic}, establishes that the WM learning algorithm can be exploited by a single strategic deviant to guarantee that every stochastically stable action profile maximizes their utility. 

\begin{example}[Exploiting the WM Learning Algorithm]\label{ex:strategic}
    Consider again the Prisoner's Dilemma from Example \ref{ex:stubborn-welfare}:
    \begin{center}
        \renewcommand{\arraystretch}{1.4}\setlength{\tabcolsep}{4pt}
        \begin{tabular}{c|>{\centering\arraybackslash}p{1.0cm}|>{\centering\arraybackslash}p{1.0cm}|}
            \multicolumn{1}{c}{} & \multicolumn{1}{c}{$A$} & \multicolumn{1}{c}{$B$} \\
            \hhline{~|-|-|}
            $A$ & \cellcolor{blue!20}\rev{$\frac{5}{10}, \frac{5}{10}$} & \rev{$0, \frac{8}{10}$} \\
            \hhline{~|-|-|}
            $B$ & \cellcolor{red!20}\rev{$\frac{8}{10}, 0$} & \rev{$\frac{2}{10}, \frac{2}{10}$} \\
            \hhline{~|-|-|}
        \end{tabular}
    \end{center}
    Suppose now that player 2 follows $\texttt{WM}(1, 2)$ rather than $\texttt{WM}(1/2, 2)$. If player 1 also follows $\texttt{WM}(1, 2)$, then the only stochastically stable action profile is again $(A, A)$ (shaded blue). Moreover, player 1 cannot make $(B, A)$ stochastically stable by following $\texttt{WM}(b_1, c_1)$ for any $b_1 \in (0, 1]$ and $c_1 \geq 2$, nor can they make $(B, A)$ stochastically stable by being stubborn with $a_1^{\rm stb} = B$. However, suppose that player 1 follows $\texttt{UM}(\Phi, \Psi)$ with
    \begin{equation*}
        \Phi(u) = \frac{2}{1 + e^{7 - 10u}} \; \text{ and } \; \Psi(u) = \frac{1}{1 + e^{10u - 7}}.
    \end{equation*}
    The functions $\Phi$ and $\Psi$ are increasing and decreasing with images in $(0, 2)$ and $(0, 1)$, respectively. Moreover, $u_1^* = 0.8$ and $\Delta_1 = 0.2$, so
        $\Phi(0.8) - \Psi(0.8) \approx 1.19 > 1 = N - 1$
    and
        $\Phi(0.6) - \Psi(0.6) \approx -0.19 < 0$.
    By Proposition \ref{prop:strategic}, the only stochastically stable action profile is $(B, A)$ (shaded red).
\end{example}

Propositions \ref{prop:heterogeneous-params} and \ref{prop:strategic} demonstrate that the WM learning algorithm specifically can be exploited. A natural question that follows is whether some other learning algorithm guarantees convergence to welfare-maximizing action profiles while giving no player an incentive to deviate from it. The final result demonstrates that this is impossible:

\begin{theorem}\label{thm:incentive}
    There exists no completely uncoupled learning algorithm $F$ such that
    \begin{enumerate}[label=(\alph*)]
        \item if every player follows $F$, the players are led to only welfare-maximizing action profiles in every finite interdependent game, and
        \item in every finite interdependent game, no player $j \in \mathcal{N}$ can strictly increase their long-run utility by unilaterally deviating from $F$ to another completely uncoupled learning algorithm $F_j$.
    \end{enumerate}
\end{theorem}

\begin{proof}
    Let $F$ satisfy (a), and fix a finite interdependent game $G$, a player $j \in \mathcal{N}$, and a strictly increasing function $\sigma: [0, 1) \to [0, 1)$. Let $F_j$ be the learning algorithm under which player $j$ follows $F$, but with each received utility $u_j$ replaced by $\sigma(u_j)$. Since $F$ is completely uncoupled, the resulting process is identical to the one induced when every player follows $F$ in the game $G'$ with utilities $U_j' = \sigma \circ U_j$ and $U_i' = U_i$ for all $i \in \mathcal{N} \setminus \{j\}$. Since $\sigma$ is injective, $G'$ is interdependent; thus, by (a), the process leads to an action profile that maximizes
    $
        \sum_{i \in \mathcal{N} \setminus \{j\}} U_i(a) + \sigma(U_j(a)).
    $
    In particular, take $G$ to be the Prisoner's Dilemma from Example \ref{ex:stubborn-welfare}, let $j = 1$, and let $\sigma$ be any strictly increasing function with $\sigma(0) = 0$, $\sigma(0.2) = 0.01$, $\sigma(0.5) = 0.02$, and $\sigma(0.8) = 0.9$. Then, $U_2 + \sigma \circ U_1$ takes the values $0.52$, $0.8$, $0.9$, and $0.21$ at $(A, A)$, $(A, B)$, $(B, A)$, and $(B, B)$, respectively, so $F_j$ leads to $(B, A)$ where player $1$ receives $U_1(B, A) = 0.8$. If player 1 follows $F$ instead, then by (a), the process leads to the unique welfare-maximizing action profile $(A, A)$ where $U_1(A, A) = 0.5$, violating (b).
\end{proof}

\noindent Taken together, the results in this section imply that no completely uncoupled welfare-maximizing learning algorithm is incentive compatible---there always exists some player who has an incentive to be a deviant.

\section{Conclusion}

As learning agents become increasingly prevalent, it is crucial to understand how differences among learners affect long-run outcomes. This work studies how a single deviant player can undermine welfare-maximizing learning, using the completely uncoupled algorithm of \cite{marden2014achieving} as a testbed. Drawing from \cite{akbar2023robustness}, we first show that a deviant who selects actions at random or always plays the same action can significantly change the long-run behavior (and welfare) of the remaining players. This weakness is not specific to this algorithm: no completely uncoupled learning algorithm is robust to such a player in every game. We then show that a deviant acting in their own interest can shift the long-run outcome in their favor. Again, this weakness is general: under any completely uncoupled welfare-maximizing learning algorithm, some player in some game benefits from switching to a different algorithm.

Together, these results show that the guarantees of completely uncoupled welfare-maximizing learning are fragile, even to a single deviant. Our analysis points to two vulnerabilities. First, existing algorithms rely heavily on interdependence to reach desirable action profiles, leaving them open to destabilization and exploitation. Second, complete uncoupledness may be too restrictive for robustness, since players cannot ``detect'' a deviant. Algorithms that do not rely on interdependence, and allow players to communicate or coordinate, may offer a path toward robustness.

\appendix

\subsection{Resistance Trees and Stochastic Stability}\label{app:stochastic-stability}

The proofs in this appendix rely on the theory of resistance trees for regular perturbed Markov processes \cite{young1993evolution}, which we briefly review following \cite{marden2014achieving}. Let $P^0$ be the transition matrix of a Markov chain on a finite state space $\mathcal{Z}$, and let $\{P^\varepsilon\}_{\varepsilon > 0}$ be a family of transition matrices on $\mathcal{Z}$. This family is a \emph{regular perturbed Markov process} if $P^\varepsilon$ is ergodic for all sufficiently small $\varepsilon > 0$, if $\lim_{\varepsilon \to 0} P^\varepsilon_{zz'} = P^0_{zz'}$ for all $z, z' \in \mathcal{Z}$, and if, whenever $P^\varepsilon_{zz'} > 0$ for some $\varepsilon > 0$, there exists a real number $r(z \to z') \geq 0$, referred to as the \emph{resistance} of the transition $z \to z'$, such that
$
    0 < \lim_{\varepsilon \to 0} \varepsilon^{-r(z \to z')} P^\varepsilon_{zz'} < \infty.
$
In particular, $r(z \to z') = 0$ if and only if $P^0_{zz'} > 0$. The resistance of a sequence of transitions $z_1 \to z_2 \to \cdots \to z_m$ is $\sum_{k=1}^{m-1} r(z_k \to z_{k+1})$.

Let $E_1, \dots, E_M$ denote the recurrence classes of $P^0$. For distinct recurrence classes $E_k$ and $E_l$, let $r(E_k \to E_l)$ denote the minimum resistance over all sequences of transitions that start in $E_k$ and end in $E_l$; we refer to $E_k \to E_l$ as an \emph{edge} with resistance $r(E_k \to E_l) > 0$. An \emph{$E_l$-tree} $T$ is a set of $M - 1$ edges such that from every $E_k \neq E_l$, there is a unique directed path in $T$ from $E_k$ to $E_l$; its resistance $R(T)$ is the sum of the resistances of its edges. The \emph{stochastic potential} $\gamma(E_l)$ is the minimum resistance over all $E_l$-trees. The following result is due to \cite{young1993evolution}:

\begin{lemma}\label{lem:young}
    Let $\{P^\varepsilon\}_{\varepsilon > 0}$ be a regular perturbed Markov process, and let $\mu^\varepsilon$ denote the stationary distribution of $P^\varepsilon$. Then $\mu^0 = \lim_{\varepsilon \to 0} \mu^\varepsilon$ exists and is a stationary distribution of $P^0$. Moreover, a state is stochastically stable if and only if it is contained in a recurrence class of $P^0$ with minimum stochastic potential.
\end{lemma}

\noindent\textbf{Resistances under the WM learning algorithm.} Consider a player $i$ following $\texttt{WM}(b_i, c_i)$. In a single timestep, the events that player $i$ may undergo have the following resistances:
\begin{itemize}
    \item If $m_i = C$, selecting an action $a_i \neq \bar{a}_i$, referred to as an \emph{experiment}, has resistance $c_i$, and selecting $\bar{a}_i$ has resistance $0$. If $m_i = D$, selecting any action has resistance $0$.
    \item If $m_i = D$, or if $m_i = C$ and $(a_i, u_i) \neq (\bar{a}_i, \bar{u}_i)$, transitioning to $(C, a_i, u_i)$, referred to as \emph{accepting} $u_i$, has resistance $b_i (1 - u_i) > 0$, and transitioning to $(D, a_i, u_i)$ has resistance $0$.
    \item If $m_i = C$ and $(a_i, u_i) = (\bar{a}_i, \bar{u}_i)$, remaining at $z_i$ has resistance $0$.
\end{itemize}
Since the players select their actions and update their states independently given the current state, the resistance of a one-step transition is the sum of the resistances of the events undergone by the individual players, and the resistance of a sequence of transitions is the sum of the resistances of all events that occur along it; a confused or stubborn player contributes zero resistance. As in \cite{marden2014achieving}, one can verify that each process considered in this paper is a regular perturbed Markov process. Two consequences will be used repeatedly:

\begin{enumerate}[label=(R\arabic*)]
    \item\label{item:R1} Consider a state in which every player following the WM learning algorithm is content and receives their baseline utility when every such player selects their baseline action. Then every transition out of this state requires at least one experiment.
    \item\label{item:R2} If a sequence of transitions starts at a state in which player $i$'s state is not $(C, \bar{a}_i, \bar{u}_i)$ and ends at a state in which it is, then the sequence contains an event in which player $i$ accepts $\bar{u}_i$, which has resistance $b_i(1 - \bar{u}_i)$. Indeed, by the state update rules, the last transition of player $i$ into $(C, \bar{a}_i, \bar{u}_i)$ is such an event.
\end{enumerate}
Finally, for each process considered in this paper, we let $\mathcal{D}$ denote the set of \emph{all-discontent} states, i.e., the states in which every player whose state is tracked by the process is discontent: every player in $\mathcal{N}$ for processes over $\mathcal{Z}$, and every player in $\mathcal{N} \setminus \{j\}$ for processes over $\mathcal{Z}_{-j}$.

\subsection{Proof of Proposition \ref{prop:confused-stubborn}}\label{app:proof-confused-stubborn}

\noindent\textbf{Part (a).} Let $P^\varepsilon$ denote the transition matrix of the process over $\mathcal{Z}_{-j}$ induced when player $j$ is confused, and let $P^0 = \lim_{\varepsilon \to 0} P^\varepsilon$. Under $P^0$, every content player selects their baseline action, every content player whose action and utility do not match their baseline becomes discontent, and every discontent player remains discontent, each with probability $1$ (since $\varepsilon^{b(1 - u_i)} \to 0$ whenever $u_i < 1$). In particular, the set of content players can only shrink under $P^0$. Let $\mathcal{R} \subseteq \mathcal{Z}_{-j}$ denote the set of states in which every player $i \in \mathcal{N} \setminus \{j\}$ is discontent and there exists $a_j \in \mathcal{A}_j$ such that $\bar{u}_i = U_i(a_j, \bar{a}_{-j})$ for all $i \in \mathcal{N} \setminus \{j\}$. We claim that $\mathcal{R}$ is the unique recurrence class of $P^0$.

First, consider any state in which every player $i \in \mathcal{N} \setminus \{j\}$ is discontent. Every action profile $a \in \mathcal{A}$ is selected with probability $1 / \lvert \mathcal{A} \rvert > 0$, after which the state of each player $i \in \mathcal{N} \setminus \{j\}$ becomes $(D, a_i, U_i(a))$ with probability $1$. Thus, under $P^0$, the process moves from any such state to a state in $\mathcal{R}$ in one step with positive probability and remains in $\mathcal{R}$.

Next, consider a state $z$ in which the set $\mathcal{S} \subseteq \mathcal{N} \setminus \{j\}$ of content players is nonempty, and let $\mathcal{I} = \mathcal{N} \setminus \mathcal{S}$. Since $j \in \mathcal{I}$ and $\mathcal{S} \neq \emptyset$, $\mathcal{I}$ is a nonempty proper subset of $\mathcal{N}$. Fix any $a_{\mathcal{I}}$; by interdependence, there exist a player $k \in \mathcal{S}$ and a choice of actions $a_{\mathcal{I}}'$ such that $U_k(a_{\mathcal{I}}', \bar{a}_{\mathcal{S}}) \neq U_k(a_{\mathcal{I}}, \bar{a}_{\mathcal{S}})$, so at least one of these two utilities differs from $\bar{u}_k$. Under $P^0$, the players in $\mathcal{S}$ select $\bar{a}_{\mathcal{S}}$, while every player in $\mathcal{I}$, each of whom is either discontent or confused, selects their action uniformly at random. Thus, with positive probability, player $k$ selects $\bar{a}_k$ but receives a utility different from $\bar{u}_k$, and becomes discontent. Repeating this argument, the process reaches a state in which every player $i \in \mathcal{N} \setminus \{j\}$ is discontent, and hence reaches $\mathcal{R}$, with probability $1$.

It follows that $P^0$ has a unique stationary distribution $\mu^0$, whose support is $\mathcal{R}$. Since $P^\varepsilon \to P^0$, taking limits in $\mu^\varepsilon = \mu^\varepsilon P^\varepsilon$ shows that every limit point of $\mu^\varepsilon$ as $\varepsilon \to 0$ is a stationary distribution of $P^0$, and hence equals $\mu^0$. Therefore, $\mu^\varepsilon \to \mu^0$, and the stochastically stable states are precisely those in $\mathcal{R}$, in which every player $i \in \mathcal{N} \setminus \{j\}$ is discontent.

\vspace{0.2cm}
\noindent\textbf{Parts (b) and (c).} Parts (b) and (c) follow from Examples \ref{ex:stubborn-welfare} and \ref{ex:stubborn-residual}, respectively; for completeness, we verify the claims made in each example. In both cases, since player $j$ always selects $a_j^{\rm stb}$, the process over $\mathcal{Z}_{-j}$ coincides with the one induced when the players $\mathcal{N} \setminus \{j\}$ follow $\texttt{WM}(b, c)$ in the game with player set $\mathcal{N} \setminus \{j\}$ and utilities $U_i(a_j^{\rm stb}, \cdot)$.

\vspace{0.1cm}
\noindent\emph{Example \ref{ex:stubborn-welfare}.} Player $2$ receives $0.5$ from $A$ and $0.8$ from $B$. Since a content player whose baseline utility differs from $U_2(A, \bar{a}_2)$ becomes discontent at the next step, the recurrence classes of $P^0$ are $z_A = (C, A, 0.5)$, $z_B = (C, B, 0.8)$, and the set $\mathcal{D}$ of states in which player $2$ is discontent. By \ref{item:R1}, every edge out of $z_A$ or $z_B$ has resistance at least $c$, and $r(z_A \to \mathcal{D}) = r(z_B \to \mathcal{D}) = c$, as player $2$ can experiment and then become discontent. By \ref{item:R2}, $r(\mathcal{D} \to z_A) = 0.5 b$ and $r(\mathcal{D} \to z_B) = 0.2 b$, where equality is attained by player $2$ selecting the corresponding action and accepting. Thus, $\{z_A \to \mathcal{D}, \mathcal{D} \to z_B\}$ is a $z_B$-tree with resistance $c + 0.2 b$. Every $\mathcal{D}$-tree contains edges out of both $z_A$ and $z_B$, so its resistance is at least $2c > c + 0.2b$. Every $z_A$-tree contains an edge out of $z_B$ and a path from $\mathcal{D}$ to $z_A$; by \ref{item:R1} and \ref{item:R2}, these contain an experiment and an acceptance of $0.5$, respectively, so its resistance is at least $c + 0.5b > c + 0.2b$. By Lemma \ref{lem:young}, $z_B$ is the unique stochastically stable state, so the unique stochastically stable action profile is $(A, B)$, which is not welfare-maximizing.

\vspace{0.1cm}
\noindent\emph{Example \ref{ex:stubborn-residual}.} With $a_3^{\rm stb} = B$, players $1$ and $2$ play the game with utilities $(U_1, U_2)$ equal to $(0, 0.1)$ at $(A, A)$, $(0, 0.9)$ at $(A, B)$, $(0.1, 0)$ at $(B, A)$, and $(0.1, 0.1)$ at $(B, B)$. Since $U_1$ depends only on $a_1$ while $U_2(A, a_2) \neq U_2(B, a_2)$ for each $a_2$, the recurrence classes of $P^0$ are the four singletons $z_{a_1 a_2}$ in which both players are content at $(a_1, a_2)$ with aligned baselines, the classes $E_{a_1}$ in which only player $1$ is content, with baseline action $a_1$, and the class $\mathcal{D}$ in which both players are discontent. By \ref{item:R1} and \ref{item:R2}, the minimum resistance of an edge out of each class is as follows:
\begin{itemize}
    \item $c$ out of each $z_{a_1 a_2}$, attained by $z_{AA} \to E_A$, $z_{AB} \to \mathcal{D}$, and $z_{BA} \to E_B$; out of $z_{BB}$, it is attained only by $z_{BB} \to E_B$ and $z_{BB} \to \mathcal{D}$, the only classes reachable with one experiment and no acceptances.
    \item $0.1b$ out of $E_A$, attained by $E_A \to z_{AB}$, and $0.9b$ out of $E_B$, attained only by $E_B \to z_{BB}$, since leaving $E_{a_1}$ requires player $1$ to experiment or player $2$ to accept $U_2(a_1, \cdot)$.
    \item $0.9b$ out of $\mathcal{D}$, attained only by $\mathcal{D} \to E_B$, since every other class has player $1$ content, which requires player $1$ to accept $U_1(a_1) \in \{0, 0.1\}$.
\end{itemize}
Thus, $\{z_{AA} \to E_A \to z_{AB} \to \mathcal{D} \to E_B \to z_{BB}, \; z_{BA} \to E_B\}$ is a $z_{BB}$-tree with resistance $3c + 1.9b$. Every tree rooted at $E_A$, $E_B$, or $\mathcal{D}$ contains edges out of all four $z_{a_1 a_2}$, so its resistance is at least $4c > 3c + 1.9b$. Every tree rooted at $y \in \{z_{AA}, z_{AB}, z_{BA}\}$ has resistance at least $3c + 1.9b$, with equality only if it contains $\mathcal{D} \to E_B$, $E_B \to z_{BB}$, and either $z_{BB} \to E_B$ or $z_{BB} \to \mathcal{D}$, which would create a cycle. By Lemma \ref{lem:young}, the unique stochastically stable action profile is $(B, B, B)$. Since $W_{-3}(B, \cdot)$ takes the values $0.1$, $0.9$, $0.1$, and $0.2$ at $(A, A)$, $(A, B)$, $(B, A)$, and $(B, B)$, respectively, $(B, B, B)$ is not residual-welfare-maximizing. \hfill $\blacksquare$

\subsection{Proof of Proposition \ref{prop:heterogeneous-params}}\label{app:proof-heterogeneous-params}

The proof follows that of \cite[Theorem 3.2]{marden2014achieving}, so we only describe the modifications. For each $i \in \mathcal{N}$, let $(b_i, c_i)$ denote the parameters of player $i$, so that $(b_i, c_i) = (b, c)$ for all $i \neq j$. Let $\mathcal{C}^0 \subseteq \mathcal{Z}$ denote the set of states in which every player is content and receives their baseline utility at the baseline action profile, i.e., $m_i = C$ and $\bar{u}_i = U_i(\bar{a})$ for all $i \in \mathcal{N}$. Since the unperturbed process does not depend on the parameters $(b_i, c_i)$, \cite[Lemma 4.1]{marden2014achieving} still applies: the recurrence classes are the all-discontent set $\mathcal{D}$ and the singletons $\{z\}$ for each $z \in \mathcal{C}^0$.

The only change is in the resistances between these classes computed in \cite{marden2014achieving}. Since an experiment by player $i$ now has resistance $c_i$, the resistance of $z \to \mathcal{D}$ becomes $\underline{c} = \min\{c, c_j\}$ for each $z \in \mathcal{C}^0$, and every edge out of $z$ has resistance at least $\underline{c}$. Since player $i$ now accepts $\bar{u}_i$ with resistance $b_i (1 - \bar{u}_i)$, the resistance of $\mathcal{D} \to z$ becomes
\begin{equation}
    \sum_{i \in \mathcal{N}} b_i (1 - \bar{u}_i) = \sum_{i \in \mathcal{N}} b_i - \sum_{i \in \mathcal{N}} b_i U_i(\bar{a}).
    \label{eq:app-D-to-z}
\end{equation}
With $c$ replaced by $\underline{c}$ and $\sum_{i}(1 - \bar{u}_i)$ replaced by \eqref{eq:app-D-to-z}, the remainder of the proof carries over. In particular, since $c_j \geq N$, we have $\underline{c} \geq N \geq \sum_i b_i$, so no state in $\mathcal{D}$ is stochastically stable.\footnote{As in \cite{marden2014achieving}, this follows by exchanging an edge $z \to w$ in a $\mathcal{D}$-tree for the edge $\mathcal{D} \to z$, where $z \in \mathcal{C}^0$ is chosen with $\bar{u}_k > 0$ for some $k$, so that the resistance of $\mathcal{D} \to z$ is strictly less than $\sum_i b_i$; such a $z$ exists unless every utility is identically zero, in which case the result is trivial.} Moreover, \cite[Lemma 4.3]{marden2014achieving} gives the stochastic potential $\underline{c} (\lvert \mathcal{C}^0 \rvert - 1) + \sum_i b_i (1 - \bar{u}_i)$ for each $z \in \mathcal{C}^0$. Since the first term is the same for all $z \in \mathcal{C}^0$, Lemma \ref{lem:young} and \eqref{eq:app-D-to-z} imply that $z$ is stochastically stable if and only if $\bar{a}$ maximizes $\sum_{i \in \mathcal{N} \setminus \{j\}} U_i + \frac{b_j}{b} U_j$; in particular, $c_j$ does not affect the stochastically stable states. \hfill $\blacksquare$

\subsection{Proof of Proposition \ref{prop:strategic}}\label{app:proof-strategic}

As with Proposition \ref{prop:heterogeneous-params}, the proof follows that of \cite[Theorem 3.2]{marden2014achieving}, so we only describe the modifications. Under the unperturbed process, player $j$ behaves exactly like a player following the WM learning algorithm: when content, they select their baseline action and become discontent if their utility changes, and when discontent, they remain discontent. Hence, \cite[Lemma 4.1]{marden2014achieving} still applies, and the recurrence classes are again $\mathcal{D}$ and the singletons $\{z\}$ for each $z \in \mathcal{C}^0$, with $\mathcal{C}^0$ as defined in Appendix~\ref{app:proof-heterogeneous-params}. Since each $z \in \mathcal{C}^0$ is uniquely determined by its baseline action profile $\bar{a}$, we augment this notation for the remainder of the proof by writing $z_a$ for the unique state in $\mathcal{C}^0$ with $\bar{a} = a$, so that $\mathcal{C}^0 = \{z_a : a \in \mathcal{A}\}$.

The resistances between these classes change as follows. Since player $j$ never experiments, leaving $z_a$ requires either an experiment by some player $i \neq j$, which has resistance $c$, or player $j$ becoming discontent while receiving $U_j(a)$, which has resistance $\Phi(U_j(a))$. Since $\Phi < N \leq c$, the resistance of $z_a \to \mathcal{D}$ is $\Phi(U_j(a))$, and every edge out of $z_a$ has at least this resistance. Moreover, since the last transition of player $j$ into $(C, a_j, U_j(a))$ has resistance $\Psi(U_j(a))$, \ref{item:R2} implies that the resistance of $\mathcal{D} \to z_a$ is
\begin{equation*}
    \Psi(U_j(a)) + \sum_{i \in \mathcal{N} \setminus \{j\}} b (1 - U_i(a)).
\end{equation*}
With the resistance $c$ of each edge $z_{a'} \to \mathcal{D}$ replaced by $\Phi(U_j(a'))$, the proof of \cite[Lemma 4.3]{marden2014achieving} carries over unchanged, since every edge out of $z_{a'}$ still has resistance at least that of $z_{a'} \to \mathcal{D}$. Letting $K = \sum_{a' \in \mathcal{A}} \Phi(U_j(a'))$ and $g = \Phi - \Psi$, this gives
\begin{equation*}
    \gamma(z_a) = K - g(U_j(a)) + \sum_{i \in \mathcal{N} \setminus \{j\}} b (1 - U_i(a)).
\end{equation*}
Similarly, every $\mathcal{D}$-tree contains an edge out of each $z_{a'}$, and the tree consisting of the edges $z_{a'} \to \mathcal{D}$ attains this bound, so $\gamma(\mathcal{D}) = K$.

Now, $g$ is increasing, since $\Phi$ is increasing and $\Psi$ is decreasing. Let $a^*$ and $a$ be action profiles with $U_j(a^*) = u_j^*$ and $U_j(a) < u_j^*$, so that $U_j(a) \leq u_j^* - \Delta_j$. Since $b \leq 1$ and $U_i \geq 0$ for all $i$, \eqref{eq:GH-cond-1} and \eqref{eq:GH-cond-2} give
\begin{align*}
    \gamma(z_{a^*}) - K &\leq (N - 1) - g(u_j^*) < 0, \\
    \gamma(z_a) - K &\geq - g(u_j^* - \Delta_j) > 0.
\end{align*}
Thus, $\gamma(z_{a^*}) < \gamma(\mathcal{D}) < \gamma(z_a)$, so the recurrence classes with minimum stochastic potential are exactly the states $z_{a^*}$ with $U_j(a^*) = u_j^*$. By Lemma \ref{lem:young}, every stochastically stable state is of this form. \hfill $\blacksquare$

\bibliographystyle{IEEEtran}
\bibliography{ref}

@article{marden2014achieving,
  title={Achieving {Pareto} optimality through distributed learning},
  author={Marden, Jason R and Young, H Peyton and Pao, Lucy Y},
  journal={SIAM Journal on Control and Optimization},
  volume={52},
  number={5},
  pages={2753--2770},
  year={2014},
  publisher={SIAM}
}

@article{marden2009payoff,
  title={Payoff-based dynamics for multiplayer weakly acyclic games},
  author={Marden, Jason R and Young, H Peyton and Arslan, G{\"u}rdal and Shamma, Jeff S},
  journal={SIAM Journal on Control and Optimization},
  volume={48},
  number={1},
  pages={373--396},
  year={2009},
  publisher={SIAM}
}

@article{pradelski2012learning,
  title={Learning efficient {Nash} equilibria in distributed systems},
  author={Pradelski, Bary SR and Young, H Peyton},
  journal={Games and Economic Behavior},
  volume={75},
  number={2},
  pages={882--897},
  year={2012},
  publisher={Elsevier}
}

@inproceedings{rajab2026log,
  title={Log-linear learning for coordination with heterogeneous bounded rationalities},
  author={Rajab, Fat-Hy O and Anubi, Olugbenga M and Vasconcelos, Marcos M},
  booktitle={2026 American Control Conference (ACC)},
  pages={1--6},
  year={2026},
  organization={IEEE}
}

@article{akbar2023robustness,
  title={Robustness of learning in games with heterogeneous players},
  author={Akbar, Aqsa Shehzadi and Jaleel, Hassan and Abbas, Waseem and Shamma, Jeff S},
  journal={IEEE Transactions on Automatic Control},
  volume={68},
  number={3},
  pages={1553--1567},
  year={2023},
  doi={10.1109/TAC.2022.3166717},
  publisher={IEEE}
}

@inproceedings{jaleel2019robustness,
  title={Robustness of stochastic learning dynamics to player heterogeneity in games},
  author={Jaleel, Hassan and Abbas, Waseem and Shamma, Jeff S},
  booktitle={2019 IEEE 58th Conference on Decision and Control (CDC)},
  pages={5002--5007},
  year={2019},
  organization={IEEE}
}

@inproceedings{shah2024learning,
  title={Learning optimal stable matches in decentralized markets with unknown preferences},
  author={Shah, Vade and Ferguson, Bryce L and Marden, Jason R},
  booktitle={2024 IEEE 63rd Conference on Decision and Control (CDC)},
  pages={787--792},
  year={2024},
  organization={IEEE}
}

@inproceedings{shah2025trial,
  title={Trial-and-error learning in decentralized matching markets},
  author={Shah, Vade and Ferguson, Bryce L and Marden, Jason R},
  booktitle={2025 IEEE 64th Conference on Decision and Control (CDC)},
  pages={4936--4941},
  year={2025},
  organization={IEEE}
}

@inproceedings{menon2014collaborative,
  title={Collaborative extremum seeking for welfare optimization},
  author={Menon, Anup and Baras, John S},
  booktitle={53rd IEEE Conference on Decision and Control (CDC)},
  pages={346--351},
  year={2014},
  organization={IEEE}
}

@article{young2009learning,
  title={Learning by trial and error},
  author={Young, H Peyton},
  journal={Games and Economic Behavior},
  volume={65},
  number={2},
  pages={626--643},
  year={2009},
  publisher={Elsevier}
}

@article{foster2006regret,
  title={Regret testing: Learning to play {Nash} equilibrium without knowing you have an opponent},
  author={Foster, Dean P and Young, H Peyton},
  journal={Theoretical Economics},
  volume={1},
  number={3},
  pages={341--367},
  year={2006}
}

@book{fudenberg1998theory,
  title={The Theory of Learning in Games},
  author={Fudenberg, Drew and Levine, David K},
  year={1998},
  publisher={MIT Press}
}

@incollection{marden2015game,
  title={Game theory and distributed control},
  author={Marden, Jason R and Shamma, Jeff S},
  booktitle={Handbook of Game Theory with Economic Applications},
  volume={4},
  pages={861--899},
  year={2015},
  publisher={Elsevier}
}

@article{hart2003uncoupled,
  title={Uncoupled dynamics do not lead to {Nash} equilibrium},
  author={Hart, Sergiu and Mas-Colell, Andreu},
  journal={American Economic Review},
  volume={93},
  number={5},
  pages={1830--1836},
  year={2003},
  publisher={American Economic Association}
}

@article{babichenko2012completely,
  title={Completely uncoupled dynamics and {Nash} equilibria},
  author={Babichenko, Yakov},
  journal={Games and Economic Behavior},
  volume={76},
  number={1},
  pages={1--14},
  year={2012},
  publisher={Elsevier}
}

@article{young1993evolution,
  title={The evolution of conventions},
  author={Young, H Peyton},
  journal={Econometrica},
  volume={61},
  number={1},
  pages={57--84},
  year={1993}
}

@inproceedings{deng2019strategizing,
  title={Strategizing against no-regret learners},
  author={Deng, Yuan and Schneider, Jon and Sivan, Balasubramanian},
  booktitle={Advances in Neural Information Processing Systems},
  volume={32},
  year={2019}
}

@article{marden2013model,
  title={A model-free approach to wind farm control using game theoretic methods},
  author={Marden, Jason R and Ruben, Shalom D and Pao, Lucy Y},
  journal={IEEE Transactions on Control Systems Technology},
  volume={21},
  number={4},
  pages={1207--1214},
  year={2013},
  publisher={IEEE}
}

@article{roughgarden2002bad,
  title={How bad is selfish routing?},
  author={Roughgarden, Tim and Tardos, {\'E}va},
  journal={Journal of the ACM},
  volume={49},
  number={2},
  pages={236--259},
  year={2002}
}

@article{zhu2013distributed,
  title={Distributed coverage games for energy-aware mobile sensor networks},
  author={Zhu, Minghui and Mart{\'\i}nez, Sonia},
  journal={SIAM Journal on Control and Optimization},
  volume={51},
  number={1},
  pages={1--27},
  year={2013},
  publisher={SIAM}
}

@article{porteagel2020wind,
  title={Wind-turbine and wind-farm flows: A review},
  author={Port{\'e}-Agel, Fernando and Bastankhah, Majid and Shamsoddin, Sina},
  journal={Boundary-Layer Meteorology},
  volume={174},
  number={1},
  pages={1--59},
  year={2020},
  publisher={Springer}
}

@inproceedings{boersma2017tutorial,
  title={A tutorial on control-oriented modeling and control of wind farms},
  author={Boersma, Sjoerd and Doekemeijer, Bart and Gebraad, Pieter and Fleming, Paul A and Annoni, Jennifer and Scholbrock, Andrew K and Frederik, Joeri and van Wingerden, Jan-Willem},
  booktitle={2017 American Control Conference (ACC)},
  pages={1--18},
  year={2017},
  organization={IEEE}
}

\end{document}